\documentclass{IEEEtran4PSCC}
\ifCLASSINFOpdf
   \usepackage[pdftex]{graphicx}
\else
   \usepackage[dvips]{graphicx}
\fi
\usepackage[cmex10]{amsmath}
\makeatletter
\let\old@ps@headings\ps@headings
\let\old@ps@IEEEtitlepagestyle\ps@IEEEtitlepagestyle
\def\psccfooter#1{%
    \def\ps@headings{%
        \old@ps@headings%
        \def\@oddfoot{\strut\hfill#1\hfill\strut}%
        \def\@evenfoot{\strut\hfill#1\hfill\strut}%
    }%
    \def\ps@IEEEtitlepagestyle{%
        \old@ps@IEEEtitlepagestyle%
        \def\@oddfoot{\strut\hfill#1\hfill\strut}%
        \def\@evenfoot{\strut\hfill#1\hfill\strut}%
    }%
    \ps@headings%
}
\makeatother

\psccfooter{%
        \parbox{\textwidth}{\hrulefill \\ \small{24th Power Systems Computation Conference} \hfill \begin{minipage}{0.2\textwidth}\centering \vspace*{4pt} \includegraphics[scale=0.06]{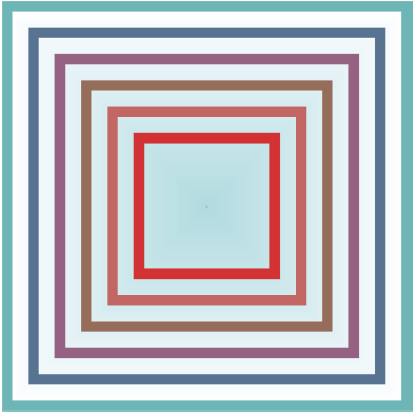}\\\small{PSCC 2026} \end{minipage} \hfill \small{Limassol, Cyprus --- June 8-12, 2026}}%
}
\usepackage{amsmath,mathtools,amsfonts,amssymb}
\usepackage{amsthm}

\usepackage{array}
\usepackage{xcolor}
\usepackage[caption=false,font=normalsize,labelfont=sf,textfont=sf]{subfig}
\usepackage{textcomp}
\usepackage{stfloats}
\usepackage{url}
\usepackage{verbatim}
\usepackage{graphicx}
\usepackage{cuted}
\usepackage{float} 
\def\BibTeX{{\rm B\kern-.05em{\sc i\kern-.025em b}\kern-.08em
    T\kern-.1667em\lower.7ex\hbox{E}\kern-.125emX}}
\usepackage{balance}
\usepackage{booktabs} % To thicken table lines
\usepackage{caption}                         % 
\usepackage{algpseudocode}
\newtheoremstyle{sltheorem}
{}                % Space above
{}                % Space below
{}        % Theorem body font % (default is "\upshape")
{10pt}                % Indent amount
{\bfseries}       % Theorem head font % (default is \mdseries)
{:}               % Punctuation after theorem head % default: no punctuation
{ }               % Space after theorem head
{}                % Theorem head spec
\theoremstyle{sltheorem}

\newtheorem{definition}{Definition}

\newtheorem{proposition}{Proposition}

\begin{document}

\title{Comparative Assessment of Look-Ahead Economic Dispatch and Ramp Products for Grid Flexibility}

\author{
Qian Zhang, Le Xie\\
School of Engineering and Applied Sciences \\
Harvard University, Allston, MA, USA \\ 
qianzhang@g.harvard.edu, xie@seas.harvard.edu
\and
Long Zhao, Congcong Wang \\
Midcontinent ISO\\
Carmel, IN, USA \\
\{lzhao, cwang\}@misoenergy.org
\thanks{This work is supported by....}}

\maketitle

\begin{abstract}
High renewable penetration increases the frequency and magnitude of net-load ramps, stressing real-time flexibility. Two commonly deployed remedies are \emph{look-ahead economic dispatch} (LAED) and \emph{ramp products} (RPs), yet their operational equivalence under the industry-standard rolling-window dispatch implementation is not well understood. This paper develops linear optimization models for multi-interval LAED and RP-based co-optimization, and proves that an \emph{enhanced} RP formulation can match LAED's dispatch feasible region at a single time step when additional intertemporal deliverability constraints are enforced. We then show that this equivalence does not generally persist under rolling-window operation because LAED and RP formulations optimize different intertemporal objectives, leading to divergent end-of-window states. Using different test systems under stressed ramping conditions and multiple load levels, we show LAED achieves similar or lower load shedding than RP implementations with the same look-ahead horizon, with the most pronounced differences under high-load, ramp-limited conditions. The study highlights the limitations of current ramp product implementations and suggests enhancements, such as introducing more mid-duration RPs.
\end{abstract}


% Use this to place sponsorships
\thanksto{\noindent Submitted to the 24th Power Systems Computation Conference (PSCC 2026).}

\section{Introduction}

The power grid is undergoing a profound transformation driven by rising renewable energy penetration and shifting demand patterns. While these changes bring substantial benefits, they also increase the frequency and magnitude of net-load ramps, creating operational challenges for maintaining security, particularly in managing net demand \emph{variation} and \emph{uncertainty} \cite{miso2022ramp}. To address these challenges, many U.S. Independent System Operators (ISOs) have adopted Look-Ahead Economic Dispatch (LAED) and/or Ramp Products (RPs)\footnote{Different ISOs use varying terms for these products, such as \emph{Flexible Ramping Products} or \emph{Ramping Capability Products}. For simplicity, we use the general term \emph{Ramp Products} throughout this paper.} in their dispatch policies, enhancing system security against future fluctuations and uncertainties \cite{wang2016ramp,caisolaed}. By anticipating system needs across multiple intervals, look-ahead approaches can proactively secure ramping capability beyond what is achievable with traditional single-interval dispatch.

LAED manages net demand \emph{variation} by extending conventional single-interval economic dispatch to a multi-interval optimization problem. In the past decade, extensive research has improved LAED performance by incorporating uncertainty through robust \cite{lorca2014adaptive}, distributionally robust \cite{Poolla}, stochastic \cite{gu2016stochastic}, and chance-constrained formulations \cite{modarresi2018scenario}, as well as by improving forecasting and scenario generation \cite{Xie2014,zhangefficient,zhangacc}. Parallel efforts have explored market mechanisms to support multi-interval dispatch in deregulated settings \cite{Anupam, guo2021pricing}.

Ramp Products have emerged as an alternative approach by introducing new reserve products to address \emph{variation} and/or \emph{uncertainty} in net demand. Initially introduced by the industry, RPs have been adopted by multiple U.S. ISOs \cite{caiso2010integration}. Prior work has proposed multi-duration RP designs \cite{wang2016ramp}, multi-timing scheduling concepts \cite{marneris2015integrated}, and market refinements such as penalty-based approaches \cite{wang2014flexible} and demand-curve adjustments \cite{chen2022addressing}. For a more comprehensive summary and development history of RPs, refer to \cite{wang2016enhancing}.

Although both LAED and RPs are widely used, prior work largely studies them in isolation, leaving a key operational question insufficiently resolved: \emph{when can RP-based procurement replicate the intertemporal deliverability and dispatch feasibility achieved by LAED, and when does the rolling-window implementation used in real-time operations cause the two approaches to diverge?} This question matters directly for operators choosing between relying more heavily on multi-interval dispatch versus product procurement, and for determining which RP durations are most valuable under sustained ramping conditions.

To address this gap, this paper compares LAED and RPs under net demand \emph{variation} assuming perfect forecasting to isolate the intertemporal ramp-deliverability mechanism. We (i) develop a unified optimization-based modeling framework for LAED and RP co-optimization, (ii) identify missing deliverability features in incomplete RP implementations and propose an enhanced RP formulation with additional intertemporal constraints, (iii) prove that the enhanced RP formulation can match LAED's feasible region at a single time step with the same look-ahead horizon, and (iv) demonstrate via rolling-window analysis and case studies that this equivalence does not generally persist over time due to differing intertemporal objectives and evolving initial conditions.

Ramp capability can be supplied by a broad set of resources beyond conventional thermal units, including battery energy storage systems (BESS), pumped-storage hydropower (PSH), flexible demand response, and aggregated electric vehicles (EVs). While the modeling in this paper abstracts away resource-specific dynamics, the insights on ramp deliverability and product duration selection directly inform how such resources could be procured and compensated through ramp product design. The paper is organized as follows: Section II presents the mathematical modeling of LAED and RPs. Section III discusses limitations of the current RPs and introduces enhanced formulations. Section IV introduces rolling-window dispatch and analyzes feasibility in rolling windows. Section V presents case studies using both 2-generator and 10-generator systems. Section VI concludes with key findings and future work.

\section{Modeling of LAED and Ramp Product}
\subsection{System Description}
Let $\mathcal{N} \doteq \{G_1, G_2, ..., G_N\}$ denote the set of dispatchable resources in the system. For each resource $G_i \in \mathcal{N}$, let $\overline{g}_i$ and $\underline{g}_i$ represent the maximum and minimum power capacities, respectively. Similarly, let $\overline{r}_i$ and $\underline{r}_i$ denote the ramp-up and ramp-down capabilities.

Consider a discrete time horizon $\mathcal{T} := \{1, 2, ..., T\}$. Given a net load profile $\{d(t): t \in \mathcal{T}\}$, the objective is to determine a feasible dispatch policy $\mathcal{G}$ that satisfies both economic efficiency and system security requirements. The dispatch policy must satisfy a set of hard operational constraints, and may be subject to penalties for soft constraint violations.

Let $g_i(t)$ denote the power output of unit $G_i$ at time $t$. Two hard constraints are considered in this model: the power capacity constraint \eqref{pc} and the ramping capability constraint \eqref{rc}.

\begin{subequations} \label{hardconstraints}
\begin{align}
\underline{g}_i \leq g_i(t) \leq \overline{g}_i,  \label{pc} \\
-\underline{r}_i \leq g_i(t) - g_i(t-1) \leq \overline{r}_i.  \label{rc}
\end{align}
\end{subequations}

\noindent \textit{Remark:} For energy-limited resources such as battery storage, energy capacity constraints also need to be considered. These, along with transmission network constraints, will be incorporated in the future work.

Hard constraints are determined by physical laws, while decision-makers aim to satisfy soft constraints based on economic or security considerations. Power balance constraints are typically treated as soft constraints. For instance, some studies propose adjustment policies for power imbalance, such as affine control, to manage potential violations \cite{bertsimas2012adaptive,zhangefficient}. In this paper, the \emph{non-negative} parameter $s(t)$ is introduced to represent the system-wide amounts of load shedding at time $t$, which has initial condition $s(0) = 0$. By incorporating this flexibility on the net demand side, the power balance constraint is relaxed to the following form (\ref{softpower}):

\begin{equation} \label{softpower}
d(t) - s(t) = \sum_i^N g_i(t).
\end{equation}
Renewable curtailment is influenced by various factors, such as oversupply and transmission congestion. To simplify the problem, it is assumed that the system has sufficient ramping-down resources, allowing the renewable curtailment term to be ignored. For instance, the market-clearing price for ramping down reserves has consistently been \$0 \cite{miso2022ramp}.

Let $\rho_s$ represent the penalty associated with load shedding $s(t)$. Because the value of lost load is extremely high, the Independent System Operator (ISO) will generally dispatch expensive offline fast-start resources to balance demand rather than directly shed load. In this paper, the penalty for dispatching offline fast-start resources is set to \$2500/MWh. This value is used to represent the high cost of emergency balancing actions when a ramp shortage occurs to manage sudden net load changes and the commitment and dispatch of offline fast-start units are not fast enough \cite{wang2016ramp}.

\subsection{Operational Security Loss}
With increasing renewable penetration, the net demand curve exhibits steeper ramping slopes, creating significant operational challenges for existing dispatch policies. Since security is \emph{above all else} from the system operators' perspective, we define the following metric to measure the operational security loss of \emph{any} dispatch policy.

\begin{definition} \label{reliabityloss}
[Operational Security Loss] The operational security loss of a dispatch policy $\mathcal{G}$, given the realized net load profile $\{d(t)\}$, is defined  as the total amount of shedding demand during time period $\mathcal{T}$, i.e.
\begin{equation} \label{totalloss}
L(\mathcal{G},\{d(t)\}) := \sum_{t \in \mathcal{T}}  \max\{0, d(t)-\sum_i^N g_i(t) \}.  
\end{equation}
\end{definition}

Many factors can contribute to operational security loss, and in some cases, load shedding cannot be avoided simply by adjusting the dispatch policy, such as when demand exceeds total system generation capacity. Rather than focusing on these inherent limitations, this paper aims to compare the operational security under different dispatch polices, especially LAED and RP solutions. To better understand their differences, we exclude other factors such as transmission constraints.

\subsection{Look-Ahead Economic Dispatch}
The LAED policy can be formulated as (\ref{laed}) for the dispatch decision made at time $t$:

\begin{equation} \label{laed}
\begin{aligned}
\mathcal{G}^{\text{LAED}_t\{W\}}:& \min  \sum_{\tau \in \mathcal{W}}  \{ \rho_s s(t + \tau)  + \sum_{i=1}^N  c_i g_i (t + \tau) \} \\
\text { s.t.}  \quad & \forall i \in \mathcal{N} \text{ and } \tau \in \mathcal{W}: \\
& \underline{g}_i \leq g_i(t + \tau) \leq \overline{g}_i, \\
& -\underline{r}_i \leq g_i(t + \tau) - g_i(t + \tau-1) \leq \overline{r}_i, \\
& d(t + \tau) - s(t + \tau) = \sum_i^N g_i(t + \tau), \\
& s(t + \tau) \geq 0,
\end{aligned}
\end{equation}
where decision variables are $\left\{g_i(t + \tau), s(t+ \tau)\right\}_{\tau \in \mathcal{W}}$. $c_i$ is the cost of dispatchable resource $i$, and  $\mathcal{W} := \{0,...,W\}$ is the look ahead window with size $W$.

\subsection{Ramp Product for Variation}
The LAED policy is used for handling net load \emph{variation} through rolling window multi-interval optimization. Rather than solely relying on new dispatch methods in the energy market, ramp product solutions are designed to tackle the net load \emph{variation} by introducing new reserve products and markets. 

\begin{definition}\label{rampproduct}
[Ramp Capability Product (Variation)] The $W$ duration ramp capability product means it is prepared to handle the future net load variation from $0$ to $W$ intervals.
\end{definition}

After introducing the ramp capability product, one of the challenges is determining the appropriate demand curve, as discussed in \cite{chen2022addressing}. To facilitate comparison with the LAED policy, the ramp product market is first considered under the assumption of perfect forecasting. In this scenario, the entire cleared ramp product is used to manage variations in net demand, with the demand curve represented as a vertical line at the value of the forecasted ramp. Assuming a bid of 0 for the ramp products, the ramp product market can be cleared together with the energy market through co-optimization. In this case, the clearing price of the ramp product is determined by the shadow price associated with the ramping constraint.

For example, if the system \emph{only} has one ramp product with $W$ duration, the dispatch policy at time $t$ can be formulated as (\ref{rp1}):

\begin{equation} \label{rp1}
\begin{aligned}
&\mathcal{G}^{\text{ED}_t\{W\}}: \min  \quad  \sum_{i=1}^N  c_i g_i (t) + \sum_{\tau \in \mathcal{W}}  \rho_s s(t + \tau) \\
& \text { s.t.}  \quad  \forall i \in \mathcal{N} \text{ and } \tau \in \mathcal{W}: \\
& -\underline{r}_i \leq g_i(t) - g_i(t-1) \leq \overline{r}_i, \\
& \underline{g}_i \leq g_i(t) + r_i^{+}(t+W) \leq \overline{g}_i, ~ \underline{g}_i \leq g_i(t) - r_i^{-}(t+W) \leq \overline{g}_i, \\
& 0 \leq r_i^{+}(t+W) \leq W*\overline{r}_i, \quad 0 \leq r_i^{-}(t+W) \leq W*\underline{r}_i, \\
& \sum_{i=1}^N r_i^{+}(t+W) \geq  (d(t+W) - s(t+W)) - (d(t) - s(t)), \\
& \sum_{i=1}^N r_i^{-}(t+W) \geq  (d(t) - s(t)) - (d(t+W) - s(t+W)), \\
& d(t)- s(t)  = \sum_i^N g_i(t), \quad s(t + \tau) \geq 0
\end{aligned}
\end{equation}
where decision variables are $\left\{g_i(t), s(t+ \tau)_{\tau \in \mathcal{W}}\right\}$ and $\{r_i^{+}(t+W),r_i^{-}(t+W)\}$.

\section{Missing Capabilities in Ramp Products}
\subsection{The Limits of Incomplete Ramp Products}

The key distinction between LAED and the ramp product lies in their approach to managing net load variation. LAED addresses net load variation across different time intervals through a single multi-interval optimization in a rolling-window way, while the ramp product approach relies on introducing multiple ramping products to manage net load variation.

In current U.S. markets, ramp products are often designed for short-duration events. For example, MISO’s RP procures ramping capability only 10 minutes ahead. However, a single-duration ramp product cannot adequately capture net load variation across multiple intervals. As demonstrated in \cite{wang2016ramp}, a 10-minute ramp product requires imposing ramping constraints at $t+5$ minutes to be effective.

Moreover, even with ramp products of \emph{all} possible durations, RP can not achieve the identical dispatch feasible region as provided by LAED, where two main constraints are missing to secure the ramping capability for the future.

\subsection{Lack of Ramp Increment Control}
To achieve the same feasible region as the LAED policy at time $t$, an additional ramp increment constraint is required in the ramp product solution. This constraint ensures that the quantity of long-duration ramping products is no less than that of shorter-duration ramping products. The \emph{ramp increment constraints} for both ramp-up and ramp-down reserves are as (\ref{rampincrement}), $\forall i \in \mathcal{N}, \tau \in \mathcal{W}, \tau > 0$:

\begin{equation} \label{rampincrement}
  \begin{aligned}
  & 0 \leq r_i^{+}(t+\tau) -  r_i^{+}(t+\tau-1) \leq \overline{r}_i, \\
  & 0 \leq r_i^{-}(t+\tau) -  r_i^{-}(t+\tau-1) \leq \underline{r}_i.
  \end{aligned}
\end{equation}

The introduction of the ramp increment constraint is primarily aimed at ensuring that the system maintains sufficient ramping capability during non-monotonic net load variations. For instance, in the ramp event shown in Fig. \ref{fig:increment}, the ramp-up reserve product with a duration of $W=2$ (40 MW/10min) will be less than the ramp-up reserve product with a duration of $W=1$ (80 MW/5min) if the ramp increment constraint is not applied at time 0:00.

\begin{figure}[!h]
\centering
\includegraphics[width=0.8\columnwidth]{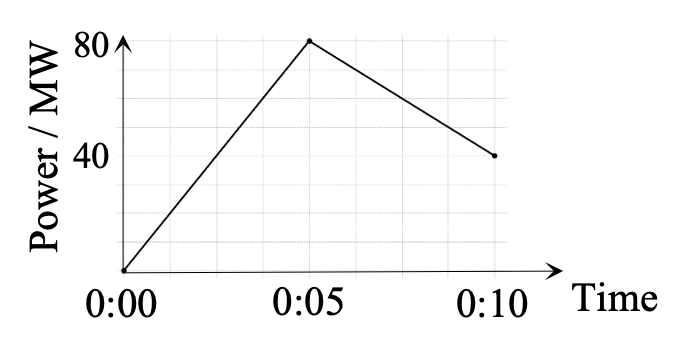}
\caption{Illustration of missing ramp increment control in ramp products}
\label{fig:increment}
\end{figure}

\subsection{Lack of Rolling Difference Enforcement}
Moreover, feasibility must hold not only between $t$ and $t+\tau$ but also across intermediate steps. The rolling difference constraints enforce that the \emph{incremental} ramp capability between successive horizons covers the incremental net-load change, which are shown as (\ref{rollingdifference}), $\forall \tau \in \mathcal{W}, \tau > 0$.

\begin{equation} \label{rollingdifference}
  \begin{aligned}
 & \sum_{i=1}^N r_i^{+}(t+\tau) -  \sum_{i=1}^N r_i^{+}(t+\tau-1)  \\
 \geq &d(t+\tau) - s(t+\tau) - (d(t+\tau-1)- s(t+\tau-1)), \\
 & \sum_{i=1}^N r_i^{-}(t+\tau) -  \sum_{i=1}^N r_i^{-}(t+\tau-1)\\
  \geq &d(t+\tau-1)- s(t+\tau-1) - (d(t+\tau )- s(t+\tau )).
  \end{aligned}
\end{equation}

The rolling difference constraint is necessary to ensure the system’s ramping capability when future net load variations are steeper than the current ones. For instance, as illustrated in Fig. \ref{fig:difference}, the ramp-up reserve products with the duration of $W=\{1, 2\}$ (10 MW/5min, 80 MW/10min) are insufficient to cover the steeper ramp during the second interval (70 MW/5min) if the rolling difference constraint is not applied.

\begin{figure}[!h]
\centering
\includegraphics[width=0.8\columnwidth]{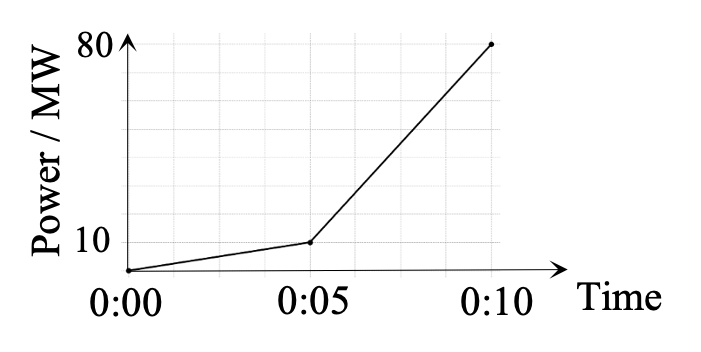}
\caption{Illustration of missing rolling difference enforcement in ramp products}
\label{fig:difference}
\end{figure}

This limitation becomes particularly impactful in systems with high renewable penetration, where net load changes can be both rapid and sustained. The current ramp product framework needs to be enhanced to account for these rolling differences and ensure that the system can maintain the required ramping capability over extended periods.

\subsection{Enhanced Ramp Products}
\begin{equation} \label{rp3}
    \begin{aligned}
    & \mathcal{G}^{\text{ED}'_t\{1,...,W\}}: \min  \quad \sum_{i=1}^N  c_i g_i (t) + \sum_{\tau \in \mathcal{W}}  \rho_s s(t + \tau) \\
    &\text { s.t.}  \quad  \forall i \in \mathcal{N} \text{ and } \tau \in \mathcal{W}: \\
    & -\underline{r}_i \leq g_i(t) - g_i(t-1) \leq \overline{r}_i, \\
    & \underline{g}_i \leq g_i(t) + r_i^{+}(t+\tau) \leq \overline{g}_i, \underline{g}_i \leq g_i(t) - r_i^{-}(t+\tau) \leq \overline{g}_i, \\
    & 0 \leq r_i^{+}(t+\tau) \leq \tau*\overline{r}_i, \quad 0 \leq r_i^{-}(t+\tau) \leq \tau*\underline{r}_i, \\
    & \sum_{i=1}^N r_i^{+}(t+\tau) \geq  d(t+\tau) - s(t+\tau) - (d(t) - s(t)), \\
    & \sum_{i=1}^N r_i^{-}(t+\tau) \geq  d(t) - s(t) - (d(t+\tau) - s(t+\tau)), \\
    & d(t) - s(t)  = \sum_i^N g_i(t), \quad s(t + \tau) \geq 0 , \\
    & \forall \tau \in \mathcal{W}, \tau > 0: \eqref{rampincrement}, \eqref{rollingdifference}
    \end{aligned}
\end{equation}

The above missing capabilities in ramp products can be enhanced by introducing ramp products with all possible durations, incorporating the ramp increment constraints \eqref{rampincrement} and the rolling difference constraints \eqref{rollingdifference}. Then, the enhanced ramp products solution can be formulated as (\ref{rp3}).

\noindent\textit{Remark}: The enhanced RP formulation introduces $O(NW)$ ramp variables and linear constraints, which is the same order as the standard $W$-step LAED formulation. The ``all durations'' construction is used to establish equivalence and to diagnose missing deliverability features in incomplete RP implementations.

In the next section, the enhanced ramp product solution will be evaluated under single-time-step and rolling-window dispatch schemes. 

\section{Rolling Window Dispatch}
In this section, we will show that the enhanced ramp products solution might achieve the same feasible region, or operational security level, as LAED at a single time step, but has no guarantee in rolling window dispatch.

\subsection{Feasibility at a Single Time Step}
\begin{proposition}\label{feasible_t}
[Equivalence of Feasible Regions at a Single Time Step] The feasible region of the \emph{enhanced} ramp products \eqref{rp3} is equivalent to that of the LAED \eqref{laed} with the same look-ahead horizon, provided that the initial generation dispatch values satisfy $g_i(t-1)$ being identical in both formulations for all $i \in \mathcal{N}$. \end{proposition}

\begin{proof}
The proof includes two parts: 1) the feasible region of ramp product solution (\ref{rp3}) is a subset of the feasible region of LAED (\ref{laed}); 2) the feasible region of LAED (\ref{laed}) is a subset of the feasible region of ramp product solution (\ref{rp3}). Let $\mathcal{R}^{\text{ED}_t'\{1,...,W\}}$ and $\mathcal{R}^{\text{LAED}_t\{W\}}$ be the feasible region of (\ref{rp3}) and (\ref{laed}), respectively.

1) $\mathcal{R}^{\text{ED}_t'\{1,...,W\}} \subseteq \mathcal{R}^{\text{LAED}_t\{W\}}$  

Given a feasible solution of LAED (\ref{laed}), denoted as ${g_i(t+\tau), s(t+\tau)}$, we define the corresponding auxiliary ramp variables as:
\begin{equation} \label{eq:ramp_definitions}
  \begin{aligned}
\text{for} \quad & \tau = 1: \\
&r_i'^{+}(t+\tau) = \max\{0, g_i(t+\tau) - g_i(t)\}, \\
&r_i'^{-}(t+\tau) = \max\{0, g_i(t) - g_i(t+\tau)\}. \\
\text{for} \quad & \tau > 1: \\
&r_i'^{+}(t+\tau) = \max\{r_i'^{+}(t+\tau -1), g_i(t+\tau) - g_i(t)\}, \\
&r_i'^{-}(t+\tau) = \max\{r_i'^{-}(t+\tau -1), g_i(t) - g_i(t+\tau)\}.
\end{aligned}
\end{equation}
Because of the original constraints in (\ref{laed}), it is easy to verify that ${r_i'^{+}(t+\tau), r_i'^{-}(t+\tau)}$ satisfies all the constraints in (\ref{rp3}). Therefore, $\mathcal{R}^{\text{ED}_t'\{1,...,W\}} \subseteq \mathcal{R}^{\text{LAED}_t\{W\}}$.

2) $ \mathcal{R}^{\text{LAED}_t\{W\}} \subseteq \mathcal{R}^{\text{ED}_t'\{1,...,W\}} $

Conversely, given a feasible solution from the ramp-product formulation (\ref{rp3}), we reconstruct the auxiliary generation dispatch sequence as:
\begin{equation} \label{eq:g_reconstruction}
  \begin{aligned}
  g_i'(t+\tau) = g_i(t) + r_i'^{+}(t+\tau) - r_i'^{-}(t+\tau).
  \end{aligned}
\end{equation}
Similarly, it is easy to verify that ${g_i'(t+\tau), s(t+\tau)}$ satisfies all the constraints in (\ref{laed}). Therefore, $\mathcal{R}^{\text{LAED}_t\{W\}} \subseteq \mathcal{R}^{\text{ED}_t'\{1,...,W\}}$.
\end{proof}

The enhanced ramp product formulation introduces additional constraints compared to the single ramp product in order to ensure future ramp deliverability. Based on Proposition \ref{feasible_t}, the enhanced ramp product and LAED may share the same feasible region, but they can yield different optimal solutions due to differences in their objective functions.

\subsection{Feasibility in Rolling Windows} \label{feasibleroll}
As illustrated in Fig. \ref{rolling}, the real-time dispatch is typically solved in a rolling-window way, which means only the dispatch decision of the first interval ($\tau = 0$) will be used as a binding instruction, while the rest $W$ intervals are advisory instructions. This is because forecasts and system conditions are updated frequently, and committing to a full multi-interval trajectory can be infeasible or inefficient after new information arrives.  
\begin{figure}[H]
\centering
\includegraphics[scale=0.55]{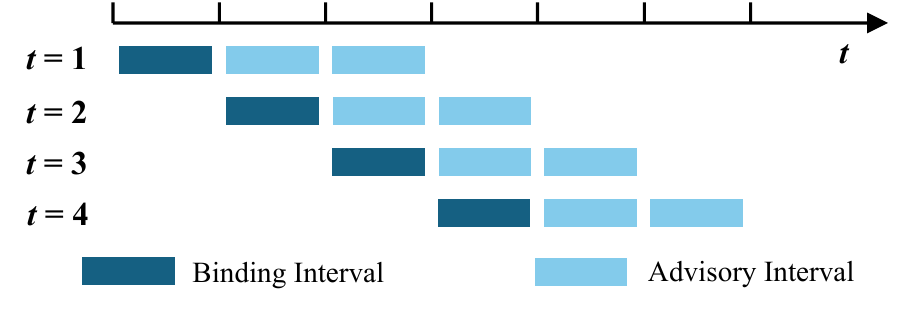}
  \caption{The rolling-window dispatch with window size $W=2$}
  \label{rolling}
\end{figure}
In rolling windows, the decision variables should be $\left\{g_{it}(t + \tau), s_t(t+ \tau)\right\}$ to represent the decision made at time $t$, but the subscript $t$ is ignored in this paper for simplicity.

The equivalence of feasible regions at a single time step relies on the assumption that the previous generator dispatch values are identical in both LAED and ramp product formulations. However, this assumption is difficult to maintain in a rolling window implementation.

\begin{proposition}\label{feasible_r}
[Feasible Regions in Rolling Windows] The feasible region of \emph{any} ramp product formulation has no guarantee of being equivalent to LAED in rolling windows.
\end{proposition}

\begin{proof}
In the LAED formulation, the objective function minimizes the \emph{total} generation cost across all intervals within the look-ahead horizon. In contrast, the ramp product formulation typically minimizes the generation cost only for the first interval of each rolling window. This structural difference can lead to divergent dispatch decisions between the two approaches, resulting in different generator output levels at the end of each window. Consequently, the initial dispatch values $g_i(t-1)$ for the subsequent time step may differ, and the feasible regions of LAED and the ramp product formulation are no longer guaranteed to be the same in a rolling window framework.    
\end{proof}

Proposition \ref{feasible_r} suggests that under the rolling window dispatch scheme, system operators may need to evaluate LAED and RPs through simulations across various scenarios and system conditions. The next section presents some empirical simulations.

\section{Case Study}
The case study begins with a 2-generator system illustrating the risk of operational security loss if dispatched with incomplete ramp products. The 10-generator system is used for comparing the performance of different dispatch polices under a longer ramp-up duration. The dispatch interval is assumed to be 5 minutes in both cases. All dispatch problems are formulated as linear programs and solved using a commercial LP solver (Gurobi). Simulation studies on larger-scale systems and detailed cost analyses are provided in our follow-up work \cite{zhang2026rampscarcity}.

\subsection{2-Generator System}
The 2-generator system used in this paper is consistent with parameters used in prior reports and studies \cite{caiso2016flexible,guo2021pricing}. In this system, G1 and G2 have the same power capacity of 500 MW. They differ in marginal costs, \$25/MWh for G1 and \$30/MWh for G2, and ramping limits of 500 MW/5min and 50 MW/5min, respectively.

The limitation of incomplete ramp product durations was first highlighted in \cite{wang2016ramp}, where additional ramp requirement constraints were proposed to ensure sufficient ramping capability is maintained in response to uncertain net load fluctuations. This section compares a system using only a 10-minute ramp product with one utilizing a 10-minute look-ahead LAED policy under short-term (less than 1 hour) simulation.

\begin{figure}[H]
  \centering
  \includegraphics[scale=0.42]{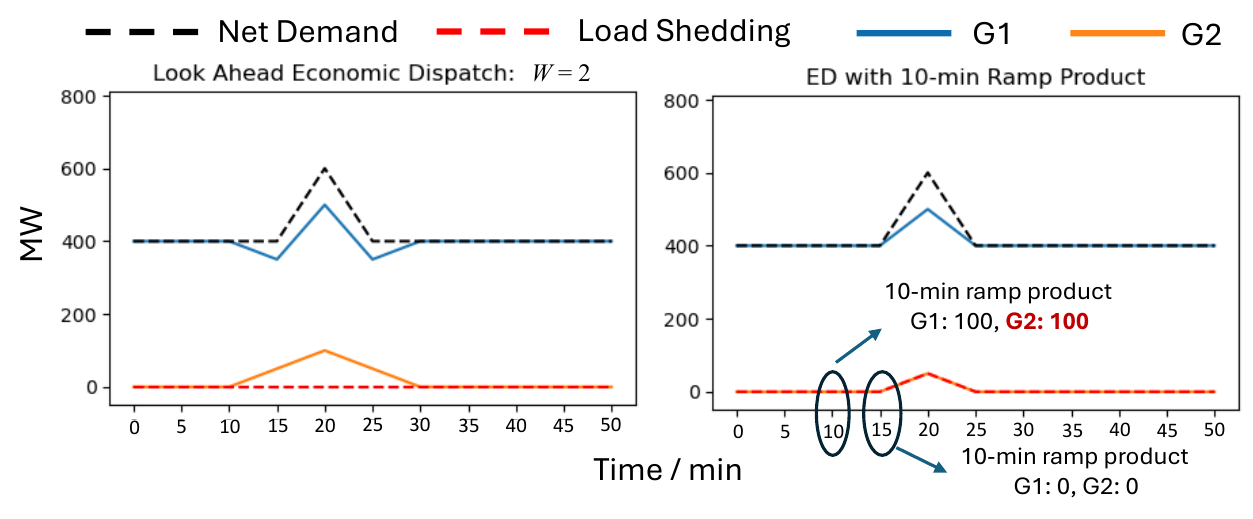}
  \caption{The comparison between the system with only a 10-min ramp product and the system with the 10-min ahead LAED (Scenario 1)}
  \label{limitation1}
\end{figure}

The absence of shorter-duration ramp products can result in operational security loss for systems relying solely on ramp products. We consider a short-term ramping spike scenario of net demand, which is shown as the black dot curve in Fig. \ref{limitation1}. The system with a 10-minute ahead LAED successfully secures the necessary ramping capability by adjusting generation at time $t=15$, proactively anticipating the ramping event. In contrast, the system relying on a 10-minute ramp product fails to avoid load shedding at $t=20$. This occurs because, at $t=15$, it only secures ramping capability for the $[15:25]$ window, which does not account for the sharper ramp needed during $[15:20]$.

This limitation is not restricted to spike events. In general, whenever the ramping requirement exhibits a decreasing trend, a system relying solely on a 10-minute ramp product may fail to maintain adequate ramping capability. As shown in Fig. \ref{limitation2}, the system with only a 10-minute ramp product fails to prevent load shedding at time $t=20$ for a similar reason.

\begin{figure}[H]
  \centering
  \includegraphics[scale=0.42]{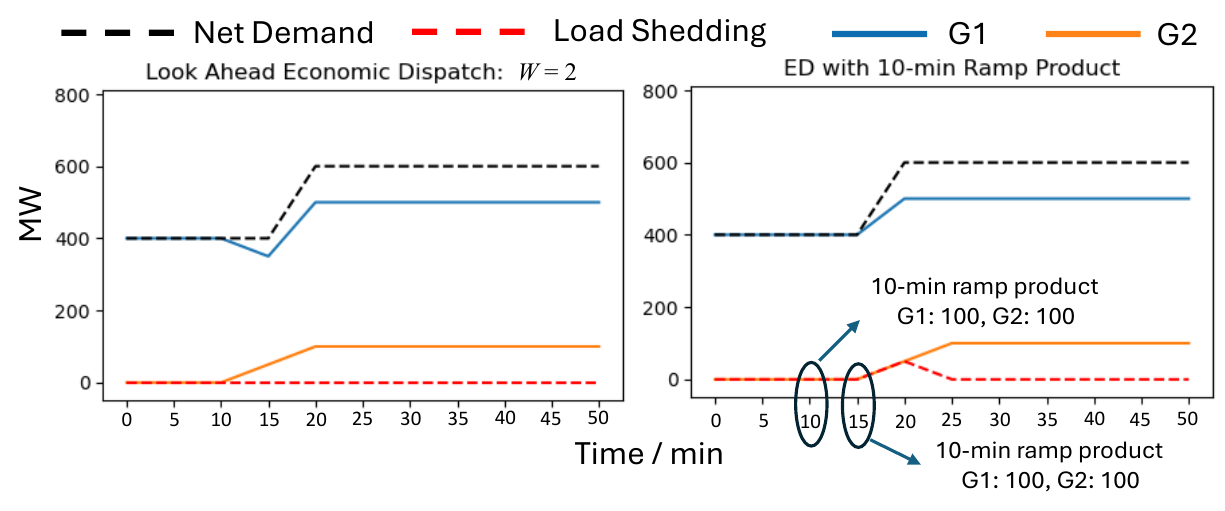}
  \caption{The comparison between the system with only 10-min ramp product and the system with 10-min ahead LAED (Scenario 2)}
  \label{limitation2}
\end{figure}

The operational security loss observed in the ramp product solution can be mitigated by introducing a 5-minute ramp product—even \emph{without} incorporating the ramp increment constraints \eqref{rampincrement} and the rolling difference constraints \eqref{rollingdifference}. For example, in the second scenario, a system equipped with both 5- and 10-minute ramp products achieves the same dispatch results as the 10-minute-ahead LAED, as illustrated in Fig. \ref{limitation3}.

\begin{figure}[H]
  \centering
  \includegraphics[scale=0.55]{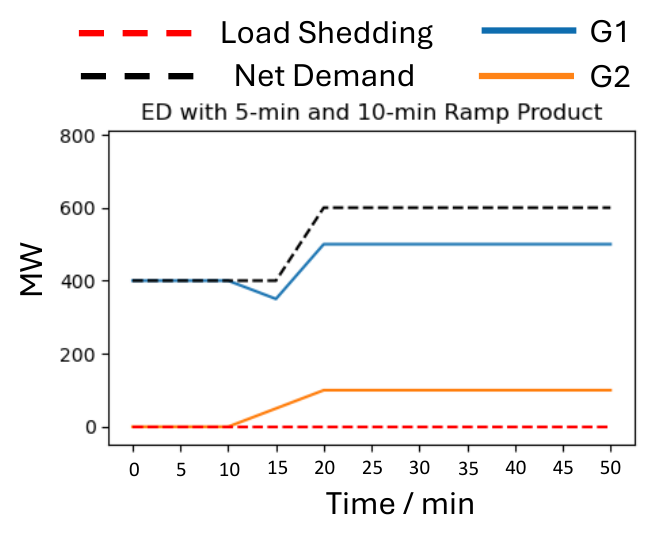}
  \caption{Dispatch results of system with 5- and 10-min ramp products (Scenario 2)}
  \label{limitation3}
\end{figure}

\subsection{10-Generator System}
The 10-generator system is used to further demonstrate the practical value of the reliability dispatch policy. System parameters, including cost, power capacity, and ramp rate limits for each generator, are provided in Table \ref{tab:generator_params}. The total system capacity is 2561 MW, with a maximum total ramp rate of 130 MW per 5-minute interval. To better illustrate the impact of ramping constraints, the ramp limits in this case study are set lower than typical real-world values.

To evaluate the operational security loss during extended ramping events, the estimated monthly average net demand profile for the MISO region in August 2032 is proportionally scaled and adopted as the test scenario \cite{MISO2023AttributesRoadmap}. This profile is illustrated by the red curve in Fig.~\ref{2032MISO}.

\begin{table}[h!]
    \centering
    \caption{Parameters for 10-Generator System}
    \begin{tabular}{c|c|c|c}
        \hline
        \textbf{Generator} & \textbf{\begin{tabular}[c]{@{}c@{}}Cost \\ (\$/MWh)\end{tabular}} & 
        \textbf{\begin{tabular}[c]{@{}c@{}}Power Capacity \\ (MW)\end{tabular}} & 
        \textbf{\begin{tabular}[c]{@{}c@{}}Ramp Limit \\ (MW/5min)\end{tabular}} \\
        \hline
        1  & 185  & 22  & 7.5  \\
        2  & 30   & 170 & 17.5 \\
        3  & 55   & 85  & 7.5  \\
        4  & 15   & 230 & 7.5  \\
        5  & 20   & 613 & 37.5 \\
        6  & 19.5   & 686 & 10 \\
        7  & 48   & 45  & 10 \\
        8  & 60   & 50  & 5  \\
        9  & 57   & 260 & 5  \\
        10 & 50   & 400 & 12.5 \\
        \hline
    \end{tabular}
    \label{tab:generator_params}
\end{table}

\begin{figure}[H]
  \centering
  \includegraphics[scale=0.48]{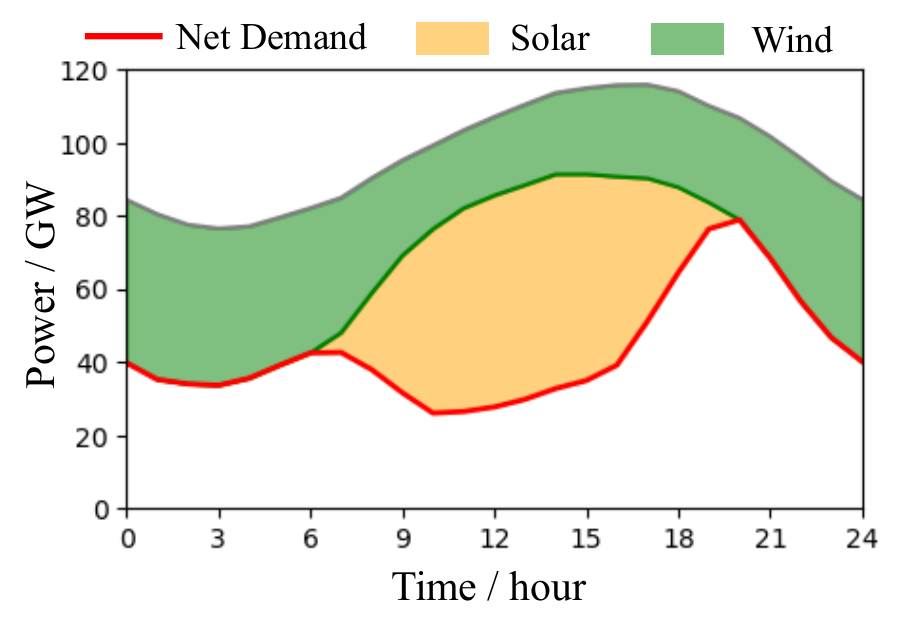}
  \caption{The estimated monthly average daily net demand for the MISO region in August 2032 \cite{MISO2023AttributesRoadmap}}
  \label{2032MISO}
\end{figure}

The load shedding results under 4 different dispatch policies with 10-min look ahead are compared in Table \ref{table_10g_10}, where the system peak demand is equal to the system capacity (2561 MW) when the mean net demand is scaled to 1395 MW. While introducing more ramp products may slightly reduce the load shedding value, the difference is not significant. In larger systems, due to the fundamental difference between LAED and ramp product solutions in the objective function (Proposition \ref{feasible_r}), it is difficult to achieve the same load shedding value across different dispatch policies under a rolling window structure. Even after incorporating the ramp increment and rolling difference constraints into the ramp product solution, the load shedding value remains slightly lower during low demand levels but is noticeably higher during high demand compared to LAED.

\begin{table}[h!]
  \caption{The load shedding values (MWh) of 10-Gen system with different 10-min ahead dispatch polices under different load levels (MW)}
  \label{table_10g_10}    
  \centering
    \begin{tabular}{c|cccc}
    \toprule
    \textbf{Load Level} & \textbf{10-min} & \textbf{10-min} & \textbf{5, 10-min} & \textbf{5, 10-min} \\
    Mean (Peak) & LAED  &RPs &  RPs &  RPs + (\ref{rampincrement}), (\ref{rollingdifference}) \\
    \midrule
    1195 (2194) & 0 & 0 & 0 & 0 \\
    1245 (2286) & 0.28 & 0.09 & \textbf{0} & \textbf{0} \\
    1270 (2332) & 5.03 & 4.83 & \textbf{4.79} & \textbf{4.79} \\
    1295 (2377) & 18.92 & 18.23 & \textbf{18.18} & \textbf{18.18} \\
    1395 (2561) & \textbf{201.08} & 214.62 & 214.62 & 214.62 \\
    \bottomrule                                                         
    \end{tabular}
  \end{table}

After extending the window size to 60 minutes ($W = 12$), the load shedding in both the system with ramp products and the system with LAED is reduced, but they are still not identical, as shown in TABLE  \ref{table_10g_60}. The system with 60-minute LAED experiences no greater load shedding value than any other dispatch method with RPs.

% \begin{figure}[H]
%   \centering
%   \includegraphics[scale=0.54]{10G_1295_2.png}
%   \caption{The comparison between the system with only 10 and 60-min ramp products and the system with 60-min LAED (10-Gen System, Mean Net Demand = 1295 MW)}
%   \label{10gen_moderate_60}
% \end{figure}

\begin{table}[h!]
  \caption{The load shedding values (MWh) of 10-Gen system with different 60-min ahead dispatch polices under different load levels (MW)}
  \label{table_10g_60}
  \centering
  \begin{tabular}{c|cccc}
  \toprule
  \textbf{Load Level} & \textbf{60-min} & \textbf{60-min} & \textbf{10, 60-min} & \textbf{10, 30, 60-min} \\
  Mean (Peak) & LAED  &RPs &  RPs & RPs  \\
  \midrule
  1195 (2194) & 0 & 0 & 0 & 0 \\
  1245 (2286) & \textbf{0} & 0.26 & 0.09 & 0.09  \\
  1270 (2332) & \textbf{2.71} & 5.39 & 5.17 & 5.17 \\
  1295 (2377) & \textbf{11.45} & 14.84 & 14.84 & 14.84  \\
  1395 (2561) & \textbf{165.14} & 180.46 & 180.46 & 177.12  \\
  \bottomrule                                                          
  \end{tabular}
\end{table}

Under the rolling window scheme, the fundamental difference in objective functions between LAED and ramp product solutions leads to non-identical operational security losses in most 10-generator system scenarios. Nevertheless, our empirical studies show that LAED generally leads to similar or significantly lower load shedding values compared to RPs when using the same look-ahead window size.

\section{Conclusion}
This study offers a comprehensive comparative assessment of LAED and RPs for managing grid flexibility. The findings suggest that LAED typically delivers similar or superior performance compared to RPs in minimizing load shedding, especially with extended look-ahead horizons. This advantage is largely attributed to LAED's capability to incorporate intertemporal constraints across different intervals, allowing for more effective anticipation and management of grid demands.

Our study also highlights the potential of RPs to be enhanced to address longer-duration ramping needs. By considering mid-duration RPs, ramp increments, and rolling difference constraints, RPs can be made more robust and capable of handling the increasing variability associated with net demand. Future work will consider forecasting error and focus on examining the differences in market mechanisms, particularly the dispatch-following incentives of LAED and RPs. Understanding these differences is crucial for designing market structures that effectively integrate these dispatch strategies, ensuring both economic efficiency and operational security.

\section*{Acknowledgment}
\emph{Disclaimer}: The views expressed in this paper are solely those of the authors and do not necessarily represent those of MISO, PJM Interconnection L.L.C. or its Board of Managers.

\bibliographystyle{IEEEtran}
\bibliography{ref}
\end{document}